\documentclass[12pt]{iopart}

\usepackage{dcolumn}
\usepackage{bm}
\usepackage{verbatim}       

\usepackage[dvips]{graphicx}
\usepackage{amsthm}
\usepackage{amssymb}
\usepackage{bm}
\usepackage{amstext}

\expandafter\let\csname equation*\endcsname\relax

\expandafter\let\csname endequation*\endcsname\relax

\usepackage{amsmath}
\usepackage{amsfonts}
\usepackage{mathrsfs}
\usepackage{cite}

\newcommand{\p}{\partial}

\newcommand{\dd}{{\rm d}}

\newcommand{\bd}{\begin{definition}}                
\newcommand{\ed}{\end{definition}}                  
\newcommand{\bc}{\begin{corollary}}                 
\newcommand{\ec}{\end{corollary}}                   
\newcommand{\bl}{\begin{lemma}}                     
\newcommand{\el}{\end{lemma}}                       
\newcommand{\bp}{\begin{proposition}}            
\newcommand{\ep}{\end{proposition}}                
\newcommand{\bere}{\begin{remark}}                  
\newcommand{\ere}{\end{remark}}                     

\newcommand{\bt}{\begin{theorem}}
\newcommand{\et}{\end{theorem}}

\newcommand{\bit}{\begin{itemize}}
\newcommand{\eit}{\end{itemize}}
\newtheorem{theorem}{Theorem}[section]
\newtheorem{corollary}[theorem]{Corollary}
\newtheorem{lemma}[theorem]{Lemma}
\newtheorem{proposition}[theorem]{Proposition}
\theoremstyle{definition}
\newtheorem{definition}[theorem]{Definition}
\theoremstyle{remark}
\newtheorem{remark}[theorem]{Remark}
\newtheorem{example}[theorem]{Example}

\begin{document}

\title{The boundary of the chronology violating set}


\author{E Minguzzi}
\address{Dipartimento di Matematica e Informatica ``U. Dini'', Universit\`a
degli Studi di Firenze, Via S. Marta 3,  I-50139 Firenze, Italy.}

\ead{ettore.minguzzi@unifi.it}

\date{}

\begin{abstract}
\noindent A sufficiently general definition for the future and past boundaries of the chronology violating region is given. In comparison to previous studies, this work does not assume  that the complement of  the chronology violating set is globally hyperbolic. The  boundary of the chronology violating set is studied and several propositions are obtained which confirm the reasonability of the definition. Some singularity theorems related to chronology violation are considered. Among the other results we prove that compactly generated horizons are compactly constructed.
\end{abstract}

\section{Introduction}

This work is devoted to the study of the boundary of the chronology
violating region of spacetime. This boundary shares some features
that are  reminiscent of the Cauchy horizon $H(S)$ of an achronal
hypersurface. In fact, in some cases in which the chronology
violation can be removed by passing to a suitable covering (that
`counts' the numbers of times a timelike curve crosses a
hypersurface $S$), this equivalence can be made manifest
\cite{maeda98}. Unfortunately, in the general case there is no such
correspondence and the boundary of the chronology violating set must
be studied for its own sake.

As today there are few results concerning the boundary of
the chronology violating set $\mathcal{C}$ , and indeed the very definition of
future and past boundary seems to be missing in the
literature. In some cases, as done by Thorne \cite{thorne93} and
Hawking \cite{hawking92}, the character and properties of the boundary $\p \mathcal{C}$  are obtained
by assuming a globally hyperbolic complement $M\backslash \bar{\mathcal{C}}$, and by  identifying the
future boundary of the chronology violating set with the past Cauchy
horizon of such chronological complement.
Of course, this
definition is not completely satisfactory as it is well posed only
for spacetimes which admit a globally hyperbolic chronological
region.

The development  of closed timelike curves in spacetimes which admit non-compact partial Cauchy hypersurfaces is fairly well understood. Hawking argued \cite{hawking92}  that if the closed timelike curves originated by the actions of an advanced civilization, then the generators of the Cauchy horizon, followed in the past direction, would enter the space region were those actions took place. This fact should be expected since the generators themselves represent the flow of the information which signals the fact that a causality violation took place. Without a breaking of the spacetime continuum those generators would have to enter a compact region, namely the future Cauchy horizon would have to be compactly generated. However,
Hawking showed that  compactly generated horizons cannot form and some well known gaps in the proof, connected with a tacit differentiability assumption on the horizon, have been recently solved \cite{minguzzi14d}, thus confirming the validity of {\em chronology protection} at the  classical level. Similar issues connected with Tipler's analysis \cite{tipler77} have also been clarified \cite{minguzzi14d}.

It has been shown that Hawking's no go theorem on the formation of timelike curves can be circumvented either by relaxing the assumption on the compact generation of the horizon or by admitting violation of the null energy condition, see \cite{ori93,grant93,olum00,ori05,ori07,dietz12} (we shall say more on this in Sect.\ \ref{cop}). The reader is warned that on this topic some imprecise or misleading statements can be repeatedly found in the literature; the  most relevant example is given by  Hawking's claims \cite[Sect.\ 3]{hawking92}
that (a)
 `absence of closed null geodesics' on compact Cauchy  horizons would be unstable, that is,  the least perturbation of the metric would cause the horizon to contain closed null geodesics; (b) `presence of closed null geodesics' would be stable. At present there is no convincing proof for these claims, and some studies seem to suggest different conclusions \cite{chrusciel94,minguzzi09d}.


In this work we study the boundary of the chronology
violating set without making  restrictive assumptions and we eventually obtain  a definition of its future and past parts. As it
happens for the concept of Cauchy horizon, the results of
this work could prove useful for the study of singularities under chronology violation. Indeed, we shall
argue that a deeper understanding of this boundary could clarify the
mutual relationship  between chronology violation and geodesic
incompleteness (i.e.\ singularities).

Let us recall that a spacetime $(M,g)$ is a connected, time-oriented
Lorentzian   manifold of
arbitrary dimension $n+1\geq 2$,  where $g \in C^k$, $k\ge 3$,  has signature $(-,+,\dots,+)$. As a
matter of notation, the boundary of a set is denoted with a dot. In
some cases in which  this notation could be ambiguous the dot is
replaced by the symbol $\partial $. The subset symbol $\subset$ is
reflexive, i.e. $X \subset X$. A set is achronal if no timelike curve joins two of its points. If $S$ is a closed achronal set, the Cauchy development $D^+(S)$ is the set of those $p\in M$ such that every past inextendible causal curve ending at $p$ intersects $S$. The Cauchy horizon is $H^+(S)=\overline{D^+(S)}\backslash I^{-}(D^+(S))$.

Let us also recall that a future
lightlike ray is a future inextendible achronal causal curve, in
particular it is a lightlike geodesic. Past lightlike rays are
defined analogously. A lightlike line is an achronal inextendible
causal curve, hence a lightlike geodesic
without conjugate points. In this work, unless otherwise specified,
all the curves will be future directed, thus, for instance, a past
lightlike ray {\em ends} at its endpoint.

The condition of absence of lightlike lines is is implied under the null genericity and the null
convergence conditions by null completeness (as these three
conditions together imply the existence of conjugate points on any
null geodesic \cite{hawking73,beem96}). Therefore, in
the study of singularity theorems it is often a good strategy to
assume the absence of lightlike lines and to look for
contradictions.

The chronology violating region $\mathcal{C}:=\{x: x\ll x\}$ is the
set formed by those points through which passes at least one closed
timelike curve. The relation $x\sim y$ if $x\ll y$ and $y \ll x$ is
an equivalence relation in $\mathcal{C}$ and, as it is well known
since the work by Carter, it splits the chronology violating region
into (open) equivalence classes denoted in square bracket,
$[x]=I^{+}(x)\cap I^{-}(x)$. Two points belonging to the same class
have the same chronological future and the same chronological past.

\section{The boundary of a chronology violating class}

In this work we are going to study the boundary of a generic
chronology violating class since the boundary of the chronology
violating region can be recovered from those. In this respect  the following result \cite[Theor.\ 4.5]{minguzzi07c}  is
worth recalling.

\begin{theorem}
Let $\dot{[x]}$ and $\dot{[y]}$ be the boundaries of the distinct
chronology violating classes  $[x]$ and $[y]$. Through every point
of $\dot{[x]}\cap \dot{[y]}$ (a set which may be empty) there passes
a lightlike line entirely contained in $\dot{[x]}\cup \dot{[y]}$. Thus, a spacetime without lightlike lines has  chronology
violating set components   having disjoint closures.
\end{theorem}


 For the proof of the next lemma see \cite[Prop.
2]{kriele89}, or the proof of \cite[Theorem 12]{minguzzi07d}.

\begin{lemma} \label{vhq}
Let $[r]$ be a chronology violating class. If $p \in \dot{[r]}$ then
through $p$ passes a future lightlike ray contained in $\dot{[r]}$
or a past lightlike ray contained in $\dot{[r]}$ (and possibly
both).
\end{lemma}

%
%
%

\begin{definition}
Let $[r]$ be a chronology violating class. The set $R_f([r])$ is
that subset of $\dot{[r]}$ which consists of the points $p$ through which
passes a future lightlike ray contained in $\dot{[r]}$. The set
$R_p([r])$ is defined analogously.
\end{definition}
%

\begin{lemma} \label{xws}
The sets $R_p([r])$ and $R_f([r])$ are closed and
$\dot{[r]}=R_p([r]) \cup R_f([r])$.
\end{lemma}

\begin{proof}
It is a consequence of the fact that  a sequence  of future
lightlike rays $\sigma_n$ of starting points $x_n \to x$ has as
limit curve a future lightlike ray of starting point $x$
\cite{minguzzi07c}, and analogously in the past case. Clearly, by
lemma \ref{vhq}, $\dot{[r]}=R_p([r]) \cup R_f([r])$.
\end{proof}

Note that it can be $R_p([r])\cap R_f([r])\ne \emptyset$ (see Fig.\ \ref{edge}).

A set $F$ is a said to be a {\em future set} if $I^{+}(F)\subset F$.
A future set is open iff $I^{+}(F)=F$. If $F$ is  future then
$J^{+}(\bar{F})\subset \bar{F}$ which implies that the closure $\bar{F}$ is future. Analogous
definitions and results hold for past sets, in particular $F$ is a
future set iff $M\backslash F$ is a past set. The boundary of a
future set is an {\em achronal boundary} \cite{beem96}.

The achronal boundary $\p I^{-}([r])$ will be particulary important
in what follows. The proof of the next result is rather standard.

\begin{proposition}
Through every point $p$ of the achronal boundary $\p I^{-}([r])$
starts a (possibly non-unique) future lightlike ray contained in $\p
I^{-}([r])$. Furthermore, if a causal curve connects two distinct
points $x$ and $y$ of $\p I^{-}([r])$ then the causal curve is
contained in $\p I^{-}([r])$ and coincides with a segment of future
lightlike ray contained in $\p I^{-}([r])$.
\end{proposition}

\begin{proof}
Let $\sigma_n$ be a timelike curve connecting $p_n\in I^{-}([r])$ to $r$, with
$p_n\to p$. By the limit curve theorem \cite{minguzzi07c} either there is a continuous
causal curve connecting $p$ to  $r$, which is impossible because
$p\notin I^{-}([r])$ or there is a future inextendible continuous
causal curve $\sigma$ contained in $\overline{I^{-}([r])}$. No point
of this curve can be contained in $I^{-}([r])$ otherwise $p\in
I^{-}([r])$ thus $\sigma\in \p I^{-}([r])$. Since $\p I^{-}([r])$ is
achronal $\sigma$ is a lightlike ray.

If the causal curve $\gamma$ connects $x$ to $y$ then between $x$
and $y$ no point of it can belong to $I^{-}([r])$ otherwise $x\in
I^{-}([r])$, a contradiction. Let $z\in \gamma \backslash\{y\}$, and
take $z'\ll z$, then $z'\ll y$ and since $I^{+}$ is open and $y \in
\p I^{-}([r])$ we have $z'\in I^{-}([r])$. Taking the limit $z'\to
z$ we obtain $z\in \overline{I^{-}([r])}$ thus $z\in \p I^{-}([r])$.
The causal curve obtained by joining $\gamma$ with the lightlike ray
starting from $y$ must be achronal as it is contained in $\p
I^{-}([r])$ and thus it is a lightlike ray.
\end{proof}

\begin{lemma} \label{hzo}
Let $[r]$ be a chronology violating class then
$I^{-}(r)=I^{-}([r])=I^{-}(\overline{[r]}\,)$ and the following sets
coincide:
\begin{itemize}
\item[(i)] $\overline{[r]}\, \backslash I^{-}([r])$,
\item[(ii)] $\dot{[r]} \backslash I^{-}([r])$,
\item[(iii)] $R_f([r]) \backslash I^{-}([r])$,
\item[(iv)] $\dot{[r]}\cap \p I^{-}([r])$.
\end{itemize}
\end{lemma}

\begin{proof}
The inclusion $I^{-}([r])\subset I^{-}(\overline{[r]}\,)$ is
obvious. The other direction follows immediately from the fact that
$I^{+}$ is open.

(i) $\Leftrightarrow$ (ii) $\Leftrightarrow$ (iii). $R_f([r])
\backslash I^{-}([r]) \subset \overline{[r]}\, \backslash
I^{-}([r])$ is trivial, $\overline{[r]}\, \backslash I^{-}([r])
\subset \dot{[r]} \backslash I^{-}([r])$ follows from $[r]\subset
I^{-}([r])$, and it remains to prove $\dot{[r]} \backslash
I^{-}([r]) \subset R_f([r]) \backslash I^{-}([r])$. Let $p\in
\dot{[r]} \backslash I^{-}([r])$, there is a sequence $p_n \in [r]$,
$p_n \to p$. Since $p_n \in [r]$ there are timelike curves
$\sigma_n$ entirely contained in $[r]$ which connect $p_n$ to $r$.
By the limit curve theorem there is either (a) a limit continuous
causal curve connecting $p$ to $r$, in which case as $[r]$ is open,
$p \in I^{-}([r])$, a contradiction, or (b) a limit future
inextendible continuous causal curve $\sigma$ starting from $p$ and
contained in $\overline{[r]}$. Actually $\sigma$ is contained in
$\dot{[r]}$ otherwise $p \in I^{-}([r])$, a contradiction. Moreover,
$\sigma$ is a future lightlike ray, otherwise there would be $q \in
\dot{[r]}\cap \sigma$, $p\ll q$ and as $I^{+}$ is open $p \in
I^{-}([r])$, a contradiction.

(ii) $\Leftrightarrow$ (iv). Let $p\in \dot{[r]}\backslash
I^{-}([r])$ and let $x\ll p$. Since $I^{+}$ is open and $p \in
\overline{[r]}$, $x \in I^{-}([r])$, and taking the limit $x \to p$
we obtain $p \in \overline{I^{-}([r])}$. But $p\notin I^{-}([r])$
thus $p\in \p I^{-}([r])$ and hence $\dot{[r]}\backslash I^{-}([r])
\subset \dot{[r]}\cap  \p I^{-}([r])$. For the converse note that if $p
\in \dot{[r]}\cap  \p I^{-}([r])$ then $p\notin I^{-}([r])$ hence
$p\in \dot{[r]}\backslash I^{-}([r])$.
\end{proof}

Let us define the sets
\[
B_f([r]):=\overline{[r]}\, \backslash
I^{-}([r]), \quad  \textrm{and} \quad B_p([r]):=\overline{[r]}\,\backslash I^{+}([r]).
\]
 Observe that $(iv)$ establishes that $B_f([r])$ is a subset of the achronal boundary $\p I^{-}([r])$ and similarly, $B_p([r])$ is a  subset of the achronal boundary $\p I^{+}([r])$.

\begin{definition}
By {\em generator} of the achronal set $A$ we mean a lightlike ray
contained in $A$.
\end{definition}
We do not impose that the generator be a maximally
extended lightlike ray contained in $A$. In other words, as a matter
of terminology, if $\sigma:[0,b)\to M$ is a generator of $A$ then
$\sigma:[a,b)\to M$, $0\le a<b$, is also a generator of $A$.

\begin{lemma}
The set $B_f([r])$ is closed, achronal, and generated by future
lightlike rays. Analogously, the set $B_p([r])$ is closed, achronal,
and generated by past lightlike rays.
\end{lemma}

\begin{proof}
Let us give the proof for $B_f([r])$, the proof for the other case
being analogous. The closure of $B_f([r])$ is immediate from the
definition. Let $p \in B_f([r])$, as $p \in R_f([r])$ there is a
future lightlike ray starting from $p$ entirely contained in
$\dot{[r]}$ and hence in $R_f([r])$. Moreover, no point of this ray
can belong to $I^{-}([r])$ otherwise $p$ would belong to
$I^{-}([r])$. We conclude that the whole ray is contained in
$B_f([r])$.

Let us come to the proof of achronality. Assume by contradiction
that there is a timelike curve $\sigma: [0,1] \to M$ whose endpoints
$p=\sigma(0)$ and $q=\sigma(1)$ belong to $B_f([r])$. There cannot
be a value of $t \in (0,1)$ such that $\sigma(t)\in \dot{[r]}$
otherwise as $I^{+}$ is open, and $p,q \in \dot{[r]}$, we would have
$r\ll \sigma(t)\ll r$, that is $\sigma(t) \in [r]$, in contradiction
with $\sigma(t)\in \dot{[r]}$. Thus either $\sigma((0,1))$ is
contained in $[r]$ or it is contained in $M\backslash
\overline{[r]}$. The former case would imply $p \in I^{-}([r])$, a
contradiction. In the latter case it is possible to find $z \in
\sigma((0,1))\cap M\backslash \overline{[r]}$, and as $p\ll z \ll q$
and $I^{+}$ is open, $r \ll z \ll r$, a contradiction.
\end{proof}

%
%
%
%

\begin{proposition} \label{chd}
Let $[r]$ be a chronology violating class, then $I^{+}([r]) \cap
\dot{[r]} \subset B_f([r])$ and $I^{-}([r]) \cap \dot{[r]} \subset
B_p([r])$. Moreover, if $p \in I^{+}([r]) \cap R_p([r])$ or
$I^{-}([r]) \cap R_f([r])$  then through $p$ passes an inextendible lightlike
geodesic contained in $\dot{[r]}$.
\end{proposition}

\begin{proof}
Let us prove the former inclusion, the latter being analogous.

Let $q \in I^{+}([r]) \cap \dot{[r]}$, we have only to prove that
$q\notin I^{-}([r])$. If it were $q\in I^{-}([r])$ then $r \ll q \ll
r$, a contradiction.
%
%
%

Let us come to the last statement. As $p \in R_p([r])$ there is a
past lightlike ray $\eta$ contained in $\dot{[r]}$ ending at $p$. As
$p \in I^{+}([r]) \cap \dot{[r]} \subset R_f([r])$, there is a
future lightlike ray $\sigma$ passing through $p$ and contained in
$\dot{[r]}$. This ray is the continuation of the past lightlike ray
$\eta$. Indeed, assume that they do not join smoothly at $p$. Take a
point $x \in I^{+}(r)\cap \eta\backslash\{p\}$ (recall that $I^+$ is open), so that, because of
the corner at $p$, $\sigma\backslash\{p\} \subset I^{+}(x)$. Again,
since $I^{+}(x)$ is open and $\sigma\subset \dot{[r]}$ we have $x
\ll r$, thus since $r \ll x$, we conclude $x\ll x$ which is
impossible as $x \in \eta \subset \dot{[r]}$. We have therefore
obtained a lightlike geodesic $\gamma=\sigma \circ \eta$ passing
through $p$ entirely contained in $\dot{[r]}$.
\end{proof}

\begin{corollary}
The following identity holds: $\dot{[r]}=B_p([r]) \cup B_f([r])$.
\end{corollary}

\begin{proof}
In a direction the inclusion is obvious, thus since
$B_p([r])=\dot{[r]}\backslash I^{+}([r])$ and $B_f([r])=\dot{[r]}
\backslash I^{-}([r])$ we have only to prove that if $p \in
\dot{[r]}$ then $p \notin I^{+}([r])$ or  $p \notin I^{-}([r])$.
Indeed, if $p$ belongs to both sets $r \ll p \ll r$, a
contradiction.
%
\end{proof}
Note that it can be $B_p([r]) \cap B_f([r])\ne \emptyset$ (see
figure \ref{edge}).
The previous results justify the following definition
\begin{definition}
The sets ${B}_f([r])$ and ${B}_p([r])$ are respectively  {\em the  future
and the past boundaries} of the chronology violating class $[r]$.
\end{definition}

The previous and the next results will prove the reasonability and the good behavior of these definitions.

\begin{proposition}
Let $[r]$ be a chronology violating class then  $I^{+}(B_f([r]))\cap
[r]=\emptyset$. Moreover,  if $B_f([r])\ne \emptyset$ then
$I^{-}(B_f([r]))=I^{-}([r])$. Analogous statements hold in the  past
case.
\end{proposition}

\begin{proof}
If there were a $p \in B_f([r])$ such that $I^{+}(p)\cap [r] \ne
\emptyset$ then $p \in I^{-}([r])$, a contradiction.

In a direction, $I^{-}(B_f([r]))\subset
I^{-}(\overline{[r]})=I^{-}([r])$. In the other direction, assume
$I^{+}([r])\cap B_f([r]) \ne \emptyset$, then there is $q \in
I^{-}(B_f([r])) \cap [r]$, hence $I^{-}([r])=I^{-}(q) \subset
I^{-}(B_f([r]))$.

The alternative $I^{+}([r])\cap B_f([r]) = \emptyset$ cannot hold,
indeed under this assumption no point of $I^{+}([r])$ would stay
outside $[r]$ as this would imply that $I^{+}([r])\cap \dot{[r]} \ne
\emptyset$ and hence because of $I^{+}([r]) \cap \dot{[r]} \subset
B_f([r])$, $I^{+}([r]) \cap B_f([r]) \ne \emptyset$. Thus the case
$I^{+}([r])\cap B_f([r]) = \emptyset$ leads to $I^{+}([r]) \subset
[r]$ and hence $I^{+}([r]) = [r]$, i.e. $[r]$ is a future set. As
$B_f([r]) \subset \dot{[r]}$, and $B_f([r])\ne \emptyset$ taken $x
\in B_f([r])$ by the property of future sets \cite[Prop.
3.7]{beem96}, $I^{+}(x) \subset [r]$ hence $x \in I^{-}([r])$ in
contradiction with the definition of $B_f([r])$.

\end{proof}

\begin{proposition}
Let $[r]$ be a chronology violating class then $B_f([r])=\dot{[r]}$
if and only if $B_p([r])=\emptyset$. Analogously,
$B_p([r])=\dot{[r]}$ if and only if $B_f([r])=\emptyset$.
\end{proposition}

\begin{proof}
The direction $B_p([r])=\emptyset \Rightarrow B_f([r])=\dot{[r]}$
follows from $\dot{[r]}=B_p([r]) \cup B_f([r])$. For the converse,
assume $B_f([r])=\dot{[r]}$ and that, by contradiction, $p\in
B_p([r])$ (hence $p \in B_p([r])\cap B_f([r])$), then $I^{-}(p)$ has
no point in $[r]$ otherwise $p \in  I^{+}([r])$ and hence $p\notin
B_p([r])$, a contradiction. Thus if $p\in B_p([r])$ then
$I^{-}(p)\cap [r]=\emptyset$. Take $q \ll p$, as $I^{+}$ is open and
$p \in \dot{[r]}$ there is a timelike curve joining $q$ to $r$. This
curve intersects $\dot{[r]}$ at some point $x$, thus $x \in
\dot{[r]}\cap I^{-}([r])$, and $x \notin B_f([r])$, a contradiction.
We conclude that $B_p([r])=\emptyset$. The proof of the time
reversed case is analogous.
\end{proof}

The definition of the edge of an achronal set can be found in
\cite[Sect. 6.5]{hawking73} or \cite[Def. 14.27]{beem96}.

\begin{definition}
Given an achronal set $S$ the edge of $S$, $\textrm{edge}(S)$, is
the set of points $q\in \bar{S}$ such that for every open set $U\ni
q$ there are $p\in I^{-}(q ,U)$, $r\in I^{+}(q,U)$, necessarily not
belonging to $S$, such that there is a timelike curve in $U$
connecting $p$ to $r$ which does not intersect $S$.
\end{definition}


It is useful to recall that $\textrm{edge}(S)$ is closed and
$\bar{S}\backslash S \subset \textrm{edge}(S)\subset \bar{S}$.

\begin{proposition}
$\textrm{edge}(B_f([r]))=\textrm{edge}(B_p([r]))$.
\end{proposition}

\begin{proof}
Let $q \in \textrm{edge}(B_f([r]))$ then for every neighborhood
$U\ni q$ there are $x,y \in U$, $x \ll q \ll y$ and a timelike curve
$\sigma$ not intersecting $B_f([r])$ connecting $x$ to $y$ entirely
contained in $U$. The point $y$ cannot belong to $\overline{[r]}$
for otherwise $q \in I^{-}([r])$ and hence $q \notin B_f([r])$
(recall that the edge of an achronal closed set belongs to the same
set), a contradiction. Every intersection point of $\sigma$ with
$\dot{[r]}$ does not belong to $B_f([r])$, and hence belongs to
$B_p([r])$. There cannot be more than one intersection point
otherwise if $z_1 \ll z_2$ are any two intersection points, $z_2 \in
I^{+}(\dot{[r]})\subset I^{+}([r])$ thus $z_2$ cannot belong to
$B_p([r])$, a contradiction. Moreover, $\sigma$ cannot enter $[r]$
otherwise, by the same argument, the next intersection point with
$\dot{[r]}$ would not belong to $B_p([r])$, a contradiction.
Let us exclude the possibility of just one intersection point between $\sigma \backslash \{x\}$ and $ \dot{[r]}$. The intersection point would belong to $B_p([r])\subset \p I^+([r])$ but not to $B_f([r])=\overline{[r]}\backslash I^{-}([r])$, thus it would belong to $\dot{[r]}\cap I^{-}([r])\subset I^{-}(r)$. Thus $\sigma$ enters $[r]$ after the intersection point, a case that we have already excluded.
We
conclude that $\sigma \backslash \{x\} \subset M\backslash
\overline{[r]}$ with possibly $x \in B_p([r])$. However, we can
redefine $x$ by slightly shortening $\sigma$ so that we can assume
$\sigma   \subset M\backslash \overline{[r]}$. It remains to prove
that $q \in B_p([r])$, from which it follows, as $\sigma$ does not
intersect $B_p([r])$, $q \in \textrm{edge}(B_p([r]))$. Assume by
contradiction, $q \notin B_p([r])$, so that $q \in
I^{+}([r])=I^{+}(r)$. Since the previous analysis can be repeated
for every $U\ni q$, we can find a sequence $x_n \notin
\overline{[r]}$, $x_n \to q$, $x_n \ll q$. As $I^{+}(r)$ is open we
can assume $x_n\gg r$, but since $x_n \ll q$ and $q \in \dot{[r]}$,
we have also $x_n \ll r$, thus $x_n \in [r]$, a contradiction. We
conclude that $\textrm{edge}(B_f([r])) \subset
\textrm{edge}(B_p([r]))$ and the other inclusion is proved
similarly.
\end{proof}

From the previous proposition it follows that
$\textrm{edge}(B_f([r])) \subset  B_f([r]) \cap B_p([r])$, however,
the reverse inclusion does not hold in general (see figure
\ref{edge}). Contrary to what happens with Cauchy horizons the
generators of the boundary do not need to reach its edge.

%
%

\begin{figure}[t]
\begin{center}
 \includegraphics[width=12cm]{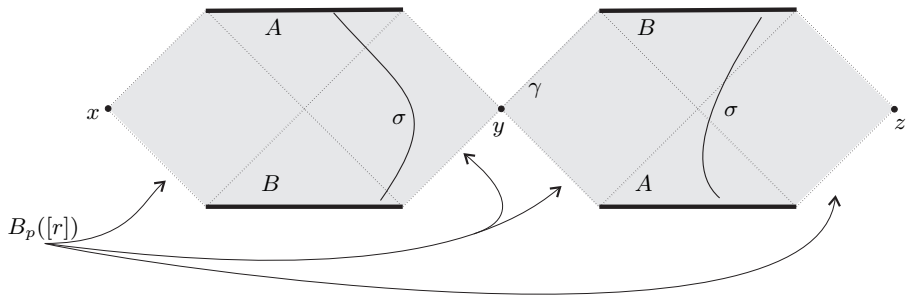}
\end{center}
\caption{Minkowski 1+1 spacetime with four spacelike segments
removed. The interior of the sides with the same label, $A$ or $B$,
have been identified. The shaded region is the only chronology
violating class $[r]$ and $\sigma$ is an example of closed timelike
curve. The points $x,y,z$ belong to $B_f([r])\cap B_p([r])$ but $x,z
\in \textrm{edge}(B_f([r]))$ while $y\notin
\textrm{edge}(B_f([r]))$. In particular, not all the generators of
$B_f([r])$ have past endpoint in $\textrm{edge}(B_f([r]))$ even if
they leave $B_f([r])$ in the past direction. The inextendible
geodesic $\gamma$ is contained in the boundary $\dot{[r]}$ but is
not achronal.} \label{edge}
\end{figure}

\begin{proposition}
The identities $B_f([r])\cap I^{+}([r])=B_f([r])\backslash
B_p([r])=\p I^{-}([r])\cap I^{+}([r])$ hold. If $p\in
B_f([r])\backslash B_p([r])$ and $\gamma:[-1,0] \to M$ is a timelike
curve such that $\gamma(0)=p$ then there is $\epsilon$,
$0<\epsilon<1$, such that $\gamma((-\epsilon,0))\subset [r]$. An
analogous past version also holds.
\end{proposition}

\begin{proof}
The first identity follows from the chain of equalities,
$B_f([r])\backslash B_p([r])=[\dot{[r]}\backslash I^{-}([r])]\cap
I^{+}([r])=B_f([r])\cap I^{+}([r])$. For the second identity, the
inclusion $B_f([r])\cap I^{+}([r])\subset \p I^{-}([r])\cap
I^{+}([r])$ is obvious. For the converse, let $p\in \p
I^{-}([r])\cap I^{+}([r])$ and let $q\ll p$ sufficiently close to
$p$ that $q\in I^{+}([r])$. Since $q\ll p$, we have $q\in
I^{-}([r])$ thus $q\in [r]$ and taking the limit $q\to p$ we obtain
$p\in \overline{[r]}$, but $p\notin I^{-}([r])\supset [r]$ thus
$p\in \dot{[r]}\cap \p I^{-}([r])\cap I^{+}([r])=B_{f}([r])\cap
I^{+}([r])$.

Let us come to the last statement. Since $I^{+}([r])$ is open and
$p\in I^{+}([r])$ there is some $\epsilon>0$ such that
$\gamma((-\epsilon,0))\subset I^{+}([r])$. But
$\gamma((-\epsilon,0))\subset I^{-}(p)\subset I^{-}([r])$ because $p
\in \overline{[r]}$, thus $\gamma((-\epsilon,0))\subset [r]$.
\end{proof}

Since $\textrm{edge}(B_f([r]))\subset B_f([r])\cap B_p([r])$ the
previous result implies the inclusion $B_f([r])\cap
I^{+}([r])\subset B_f([r])\backslash \textrm{edge}(B_f([r]))$.

\begin{proposition} \label{kqx}
$B_f([r])\backslash \textrm{edge}(B_f([r]))$ is an open set (in the
induced topology) of the achronal boundary $\p I^{-}([r])$.
An analogous past
version also holds.
\end{proposition}

\begin{proof}
Let $B=\p I^{-}([r])$ and let $q \in B_f([r])\backslash
\textrm{edge}(B_f([r]))$. We want to prove that there is a
neighborhood $U\ni q$ such that $U\cap B_f([r])=U\cap B$. By
contradiction assume not, then for every causally convex
neighborhood $U\ni q$ and $x,y \in U$, $x\ll q\ll y$, we consider
the neighborhood of $q$, $I^{+}(x)\cap I^{-}(y)$. By assumption this
neighborhood contains some point $z \in B\backslash B_f([r])$. The
timelike curve $\eta\subset U$ joining $x$ to $z$ and then $z$ to
$y$ does not intersect $B_f([r])$. Indeed, $x,y \notin B_f([r])$ as
$q\in B_f([r])$ and $B_f([r])$ is achronal. The curve $\eta$ cannot
intersect $B_f([r])$ between $x$ and $z$ because, as $z \in B$, and
$B_f([r])\subset B$ it would imply that $B$ is not achronal.
Analogously, $\eta$ cannot intersect $B_f([r])$ between $z$ and $y$
because, as $z \in B$, and $B_f([r])\subset B$ it would imply that
$B$ is not achronal. Since every point admits arbitrarily small
causally convex neighborhoods we have proved $q\in
\textrm{edge}(B_f([r]))$ a contradiction.
\end{proof}
Figure \ref{edge} shows that $B_f([r])\backslash
\textrm{edge}(B_f([r]))$ can be different from $B_f([r])\cap
I^{+}([r])$. A non-trivial problem consists in  establishing if  $B_f([r])$ can be
defined as $\overline{\p [r]\cap I^{+}(r)}$. The answer is
affirmative and shows in particular that no point of
$\textrm{edge}(B_f([r]))$ is isolated from $B_f([r])\backslash
\textrm{edge}(B_f([r]))$ or from $B_f([r])\cap I^{+}([r])$.

In the next theorem $\textrm{Int}_{\p I^{-}([r])} $ denotes the
interior with respect to the topology induced on the achronal
boundary $\p I^{-}([r])$.

\begin{theorem} \label{nzv}
The identities $B_f([r])=\overline{ \dot{[r]}\cap
I^{+}(r)}=\overline{B_f([r])\backslash \textrm{edge}(B_f([r]))}$ and
\[\textrm{Int}_{\p I^{-}([r])} B_f([r])=B_f([r])\backslash
\textrm{edge}(B_f([r]))\]  hold. Analogous past versions also hold.
\end{theorem}

\begin{proof}
Let us prove the  identity $B_f([r])=\overline{ \dot{[r]}\cap
I^{+}(r)}$. Since $\dot{[r]}\cap I^{+}(r)\subset B_f([r])$ one
direction  is obvious. For the other direction, let $p\in B_f([r])$. By lemma \ref{hzo} (iv)
$B_f([r])\subset \p I^{-}([r])$. Since $\p I^{-}([r])$ is an
achronal boundary it is possible to introduce in a neighborhood $O$
of $p$ coordinates $\{x^0,x^1,\ldots, x^{n}\}$ such that $\p/\p
x^0$ is timelike and the timelike `vertical' curves $\{x^i=cnst.\
(i=1, \ldots n)\}$ intersect $\p I^{-}([r])$ exactly once.
Furthermore in these coordinates the achronal boundary $O\cap \p
I^{-}([r])$ is expressed as the graph of a function $x^0(\{x^i, i\ne 0\})$ which
is Lipschitz \cite[Prop. 6.3.1]{hawking73}. Let $p_n\in [r]\cap O$
be a sequence such that $p_n\to p$. The timelike vertical curve
$\sigma$ passing through $p_n$ intersects $\p I^{-}([r])$ at some
point $q_n$ different from $p_n$ because $p_n\in I^{-}(r)$. It
cannot be $q_n \ll p_n$ otherwise $q_n \in I^{-}([r])$ while $q_n
\in \p I^{-}([r])$, a contradiction. Thus we have just $p_n \ll q_n$. Since $q_n
\in \p I^{-}([r])$ and $I^{+}$ is open for every  $U\ni q_n$ there
is some point $q_n'\in U\cap I^{+}(p_n)\cap I^{-}(r)\subset U\cap
I^{+}(r)\cap I^{-}(r)$ which implies $q_n'\in [r]$ and since $U$ is
arbitrary $q_n \in \dot{r}$. Furthermore, we have  $q_n\in
I^{+}(p_n)=I^{+}([r])$, $q_n \in \dot{r}  \cap I^{+}([r])$, and the
continuity of the graphing function $x^0({\bm x})$ of the achronal boundary implies $q_n \to
p$, that is $p\in \overline{\dot{r} \cap I^{+}([r])}$.

The identity $B_f([r])=\overline{B_f([r])\backslash
\textrm{edge}(B_f([r]))}$ follows from $B_f([r])=\overline{
\dot{[r]}\cap I^{+}(r)}$ using the inclusion $\dot{[r]}\cap
I^{+}(r)\subset B_f([r])\backslash \textrm{edge}(B_f([r]))\subset
B_f([r])$ proved in Prop. \ref{kqx}.

Coming to the last identity, the inclusion
\[\textrm{Int}_{\p I^{-}([r])} B_f([r])\supset B_f([r])\backslash
\textrm{edge}(B_f([r]))\] is a rephrasing of proposition \ref{kqx}.
Suppose that the reverse inclusion does not hold, then there is $p
\in  \textrm{edge}(B_f([r]))$ and an open neighborhood $U\ni p$,
such that $U\cap \p I^{-}([r])\subset B_f([r])$. However, this is
impossible because taking $r \ll p \ll q$, $q,r\in U$, they must be
connected by a timelike curve contained in $U$ which does not
intersect $B_f([r])$, but since $\p I^{-}([r])$ is edgeless and
$p\in \p I^{-}([r])$, this curve intersects $\p I^{-}([r])$ at some
point inside $U$ thus belonging to $B_f([r])$, a contradiction.

\end{proof}

\begin{corollary}
If $\textrm{edge} (B_f([r]))=\emptyset$ then  $B_f([r])$ is a
connected component of $\p I^{-}([r])$.
\end{corollary}

\begin{proof}
By theorem \ref{nzv} $\textrm{Int}_{\p I^{-}([r])}
B_f([r])=B_f([r])$, thus $B_f([r])$ is an open and closed subset of
$\p I^{-}([r])$ in the induced topology from which the thesis
follows.
\end{proof}

Theorem \ref{nzv} proves that $[r]$ is like a shell, the boundary
$\dot{[r]}$ is obtained by gluing the two $n$-dimensional
topological submanifolds $B_f([r])\backslash
\textrm{edge}(B_f([r]))$ and $B_p([r])\backslash
\textrm{edge}(B_f([r]))$ along their rims. Furthermore, these
submanifolds can touch in some points in their interior.
Nevertheless, as the next result proves, this touching region has
vanishing interior.

\begin{proposition}
The following identity holds
\[
\textrm{Int}_{\p I^{-}([r])} (B_f([r])\cap B_p([r]))=\emptyset.
\]
\end{proposition}

\begin{proof}

Let  $p\in B_f([r])\cap B_p([r])$, since $p\in
B_p([r])=\overline{\dot{[r]}\cap I^{+}([r])}$ there is a sequence
$p_n\in \dot{[r]}\cap I^{+}([r])$ such that $p_n \to p$, but $p_n
\in B_f([r])\subset\dot{[r]}$ and $p_n \in I^{+}([r])$ thus
$p_n\notin B_p([r])$ hence $p_n\in B_f([r])\backslash B_p([r])$
which proves the thesis.
\end{proof}

The next example proves that $\textrm{edge}(B_f([r]))$ is not
necessarily acausal and that in fact $\textrm{edge}(B_f([r]))$ could
be generated by inextendible lightlike lines (see figure
\ref{hawk}).
\begin{example}
Let $M=\mathbb{R}\times \mathbb{R}^2$ be endowed with the metric
\[
\dd s^2=-2\big(\cos \alpha(r) \dd t-\sin\!\alpha(r) \, r\dd \varphi
\big)\big(\sin \alpha(r)\dd t +\cos\!\alpha(r) \,r\dd \varphi\big) + \dd r^2
\]
where $(r,\varphi)$ are polar coordinates on $\mathbb{R}^2$, and
$\alpha\colon [0,+\infty)\to [0, \pi/4]$ is such that $\alpha(0)=\pi/4$
and $\alpha=0$ (only) for $r=1$, an $\dd \alpha/\dd r(0)=0$. This
metric can be obtained from the usual Minkowski 1+2 metric by
tilting the cones of an angle $\pi/4-\alpha(r)$ in the positive
$\varphi$ direction. The cones become tangent to the slices
$t=const$ at $r=1$ and then begin to tilt up again. As a result $t$
is a semi-time function, in the sense that $x\ll y \Rightarrow t(x)
< t(y)$. The curves $t=const. $, $r=1$, are closed lightlike curves
and since they are achronal they are lightlike lines.

The metric can be written in the Kaluza-Klein reduction form
\[
\dd s^2= r^2 \sin 2\alpha \Big(\dd \varphi-\frac{1}{r \tan 2\alpha}
\,\dd t\Big)^2+ \Big[-\frac{1}{\sin 2 \alpha} \, \dd t^2+  \dd r^2\Big].
\]
If we focus on sets that are rotationally invariant the causal sets
corresponding to those are obtained just considering the metric in
square brackets rather than the full metric. This is a general
feature of spacelike dimensional reduction, and rests on the fact
that the horizontal lift of a causal curve on the base is a causal
curve in the full spacetime and the projection of a causal curve of
the full spacetime is a causal curve  on the base. Furthermore, for
what concerns causality the metric in square brackets can be
multiplied by a conformal factor so that in the end the casuality is
determined by the metric $-\dd t^2+\sin 2 \alpha \dd r^2$.

The idea is  to consider the disk $S=\{x\colon t(x)=0, r(x)\le 1\}$,
represented in the reduced spacetime by the segment $[0,1]$ and
define $C^{\pm}=\{y: t(y)=\pm k\}\cap D^{\pm}(S)$. For reasons of
symmetry $C^{\pm}$ is a, possibly empty, disk but for $k$
sufficiently small $C^{\pm}$ has non-vanishing radius. The fact that
the causality can be reduced to that of a 2-dimensional spacetime,
and the fact that in 2-dimensional spacetime the geodesics do not
have conjugate points \cite[Lemma 10.45]{beem96} implies the
identity $J^{-}(C^+)\cap J^{+}(S)=D^+(S)$. Indeed both rotationally
invariant sets have a boundary described by the equation
$t(r)=\int_r^1\sqrt{\sin 2 \alpha}\, \dd r'$. In particular the
radius $R$ of $C^{\pm}$ satisfies $k=\int_R^1\sqrt{\sin 2 \alpha}\,
\dd r'$.

Our spacetime is constructed by removing $C^+$ and $C^-$ and by
identifying the interior of the lower side of the former set with the
interior of the upper side of the latter set. In this way we get a
chronology violating class $[r]$ such that $\textrm{edge}(B_f([r]))$
is the rim $\gamma$ of $S$, hence a closed achronal geodesic. In
this example the generators of $B_f([r])$ are past inextendible
lightlike geodesic which accumulate on $\textrm{edge}(B_f([r]))$
without reaching it.
\end{example}

\begin{figure}[t]
\begin{center}
 \includegraphics[width=12cm]{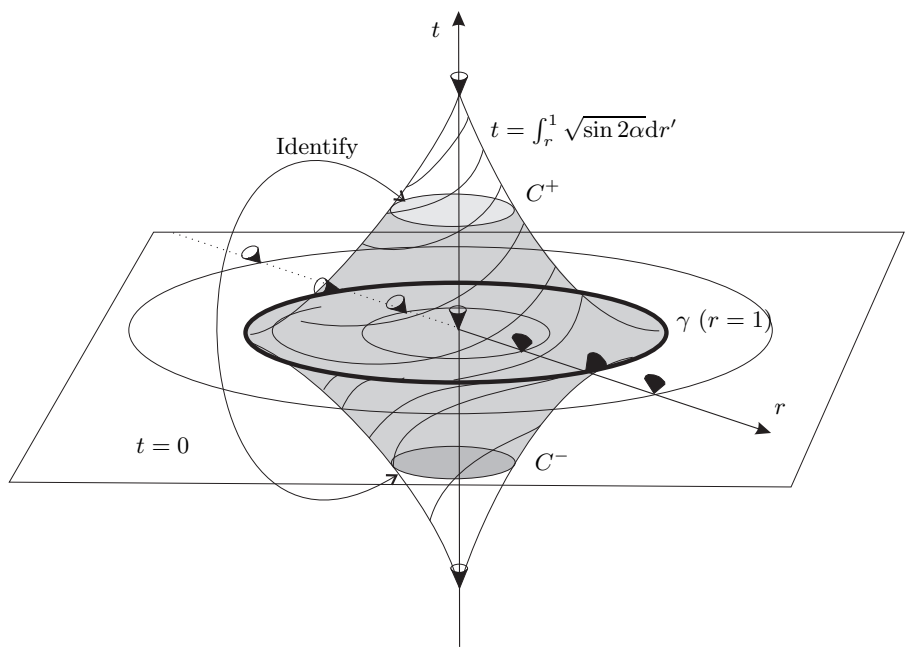}
\end{center}
\caption{The sets $C^\pm$ are removed and their sides are suitably
identified. The shaded region is the chronology violating set. The
edge of its future (or past) boundary is  the closed achronal
geodesic $\gamma$. The figure is similar to \cite[Fig.
31]{hawking73} but the cones tilt in a different way and the
generators running over the future (or past) boundary do not reach
$\gamma$.} \label{hawk}
\end{figure}

Let us investigate the causal convexity of the chronology violating set and its boundaries.

\begin{proposition} \label{jec}
Let $r\in \mathcal{C}$, the set $[r]$ is causally convex and the set
$\overline{[r]}$ is chronologically convex. Moreover, the sets
$B_{f}([r])\backslash B_p([r])$ and $B_{p}([r])\backslash B_f([r])$
are causally convex. Finally, if $\overline{[r]}$ is not causally
convex then there is an inextendible lightlike geodesic without
conjugate points which intersects both $B_p([r])$ and $B_f([r])$, in
fact it is a generator for these sets for a suitably restricted
domain of definition.
\end{proposition}

\begin{proof}
Let $x\le y \le z$ with $x,z\in [r]$. We know that $r\ll x$ and $z
\ll r$ thus $r\ll y\ll r$, that is $y \in [r]$, which proves that
$[r]$ is causally convex.

Let $x\ll y \ll z$ with $x,z\in \overline{[r]}$. Since $I^{+}$ is
open there are $x',z'\in [r]$ such that $x'\ll y\ll z'$ thus $y \in
[r]$, which proves that $\overline{[r]}$ is chronologically convex.

Let us come to the last statement. Let $x\le y \le z$ with $x,z\in
\overline{[r]}$. If $x\ll y$ or $y \ll z$ then it is easy to
construct a timelike curve connecting $x$ to $z$ which passes
arbitrarily close to $y$. Since this timelike curve is necessarily
contained in $[r]$ (because $x,z\in \overline{[r]}$ and $I^+$ is open) we get $y\in
\overline{[r]}$. We can therefore assume that $x$ is connected to
$y$ by an achronal lightlike geodesic and analogously for the pair
$y$, $z$. If the two geodesic segments do not join smoothly it is
possible again to construct, using the smoothing of the corner
argument, a timelike curve which connects $x$ to $z$ which passes
arbitrarily close to $y$. We can therefore consider the case in
which $x$ and $z$ are connected by a lightlike geodesic segment
$\gamma$ passing through $y$.

Let us consider the case $x,z\in B_{f}([r])\subset \p P$ where
$P=I^{-}([r])$. Since for every past set $J^{-}(\bar P)\subset
\bar{P}$ and $\gamma$ does not enter $P$ (otherwise $x\in
I^{-}([r])$ a contradiction) we have $\gamma\subset \p P=\p
I^{-}([r])$. Let $y'\in \gamma\backslash\{x,z\}$.

If $x\in B_{f}([r])\backslash B_{p}([r])$ then $x \in I^{+}([r])$
and it is possible to find a timelike curve $\sigma$ connecting $r$
to $z$ passing arbitrarily close to $y'$. Since $I^{+}$ is open,
$\sigma\backslash \{z\}\subset I^{-}([r])$ thus $\sigma\backslash
\{z\}\subset [r]$ and $y' \in \overline{[r]}$. Together with
$\gamma\subset \p I^{-}([r])$ this fact implies $\gamma \subset
B_{f}([r])$ and in particular $y\in B_{f}([r])$. (This case proves
also that $B_{f}([r])\backslash B_p([r])$ is causally convex indeed
it cannot be $y \in B_p([r])$ as $x\in I^{+}([r])$ and thus $y\in
 I^{+}([r])$.)

If $x\in B_{f}([r])\cap B_p([r])$ let $\sigma$ be the past lightlike
ray contained in $B_p([r])$ ending at $x$. If $\sigma$ does not join
smoothly with $\gamma$ then $\gamma\backslash\{x\}\subset
I^{+}([r])$ thus $\gamma\backslash\{x\}\subset \p I^{-}([r])\cap
I^{+}([r])=B_{f}([r])\backslash B_{p}([r])$, in particular $y\in
B_{f}([r])$. If $\sigma$ joins smoothly with $\gamma$ let us
consider a future inextendible lightlike ray $\eta$ starting from
$z$ and contained in $B_f([r])$. If $\eta$ does not join smoothly
$\gamma$ then $\gamma\backslash\{z\}\subset I^{-}([r])$ which is
impossible since $x\in B_{f}([r])$. Thus we are left with the case
in which $\gamma$ can be extended to an inextendible lightlike
geodesic which in the past direction becomes a generator of
$B_{p}([r])$ (coincident with $\sigma$) and in the future direction
becomes a generator of $B_f([r])$ (coincident with $\eta$).

The case $x,z\in B_{p}([r])$ leads to time dual results and we are
left only with the cases (i) $x\in B_{p}([r])\backslash B_{f}([r])$,
$z\in B_{f}([r])\backslash B_{p}([r])$, and (ii) $x\in
B_{f}([r])\backslash B_{p}([r])$, $z\in B_{p}([r])\backslash
B_{f}([r])$. The case (ii) cannot apply because $z\in I^{-}([r])$
which would imply $x\in I^{-}([r])$ a contradiction with $x\in
B_{f}([r])$. In case (i) let $\sigma$ be the past inextendible
lightlike ray contained in $B_p([r])$ ending at $x$ and let $\eta$
be the future inextendible lightlike ray contained in $B_f([r])$
ending at $z$. If $\sigma$ does not join smoothly with $\gamma$ then
$\gamma\backslash\{x\} \subset I^{+}([r])$ and it is possible to
find a timelike curve $\alpha$ connecting $r$ to $z$ passing
arbitrarily close to $y$. Since $I^{+}$ is open, $\alpha\backslash
\{z\}\subset I^{-}([r])$ thus $\alpha\backslash \{z\}\subset [r]$
and $y \in \overline{[r]}$. Analogously, if $\eta$ does not join
smoothly with $\gamma$ then $y'\in \overline{[r]}$. Thus also in
case (ii) we get that $\gamma$ can be extended to an inextendible
lightlike geodesic which in the past direction becomes a generator
of $B_{p}([r])$ (coincident with $\sigma$) and in the future
direction becomes a generator of $B_f([r])$ (coincident with
$\eta$).

If this geodesic contains a pair of conjugate points then by taking
a small timelike variation \cite[Prop. 4.5.12]{hawking73}, every
curve of the variation belongs to the chronology violating set and
hence $y$ belongs to the closure of the chronology violating set.
Thus if $y\notin \overline{[r]}$ the  constructed inextendible
geodesic has no pair of conjugate points.
\end{proof}
The set $B_f([r])\cap B_p([r])$  is not necessarily causally convex,
see Figure \ref{ode}.

\begin{figure}[t]
\begin{center}
 \includegraphics[width=9cm]{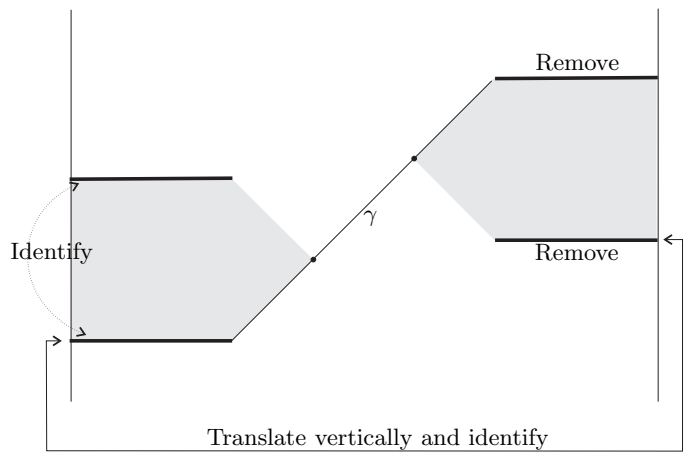}
\end{center}
\caption{The region of Minkowski 1+1 spacetime between two parallel
timelike geodesics. These timelike geodesics are identified after a suitable
vertical translation. Two spacelike segments  are removed and their
interior suitably identified (so the two horizontal segments on the right are the continuation of those on the left). The shaded region is the only
chronology violating class $[r]$. The boundary $\dot{[r]}$ is not
causally convex and there is a inextendible geodesic $\gamma$
without conjugate points which is a generator for both $B_f([r])$
and $B_p([r])$ (see Prop. \ref{jec}). The two displayed points
belong to $B_f([r])\cap B_p([r])$.} \label{ode}
\end{figure}

If we follow a generator of $B_f([r])$ in the past direction we may
suspect that  as long as the geodesic stays in $\dot{[r]}$ its
points belong to $B_f([r])$. This is false as Figure \ref{edge}
shows, however, if the geodesic does not enter $B_p([r])\backslash
B_{f}([r])$ then it is true as the next proposition proves.

\begin{proposition} \label{hzw}
Let $\gamma:[0,b) \to M$, $0<b$, be a causal curve which is a generator of $B_f([r])$ if
restricted to the domain $[a,b)$, $0\le a<b$. If
$\gamma([0,a))\subset \dot{[r]}$ and $\gamma(0)\in B_f([r])$ then
$\gamma:[0,b) \to M$, is a generator of $B_f([r])$.
\end{proposition}

\begin{proof}
Let $t\in (0,b)$, it cannot be $\gamma(t)\in B_p([r])\backslash
B_{f}([r])\subset I^{-}([r])$ for in this case $\gamma(0)\in
I^{-}([r])$, a contradiction. Thus $\gamma\subset B_{f}([r])$ and
$\gamma$ is necessarily achronal as $B_{f}([r])$ is achronal, i.e.
$\gamma$ is a generator.
\end{proof}

If we follow a generator of $B_f([r])$ in the past direction we may
suspect that the first exit point from $\dot{[r]}$ (if there is any)
should belong to $\textrm{edge}(B_f([r]))$. This is not generically
so as as the next proposition and example show. Ultimately the generators do not end on the edge as it happens for Cauchy horizons because in the present  case the inextendible direction of the generators moves away from the edge while in the Cauchy horizon case it moves towards the edge.

\begin{proposition}
Let $\gamma:(-a,b) \to M$, $0<a,b$, be a causal curve which is a generator of $B_f([r])$
if restricted to the domain $[0,b)$,  and assume that
 for every $\epsilon>0$,
$\gamma((-\epsilon,0))$ contains some point in $M\backslash
\overline{[r]}$, then $\gamma(0)\in B_f([r])\cap B_p([r])$.
\end{proposition}

\begin{proof}
As $\gamma$ is causal, $\gamma \subset \overline{I^{-}([r])}$. It cannot be $\gamma(0)\in B_f([r])\backslash B_{p}([r])$ for in
this case for sufficiently small $\epsilon$,
$\gamma((-\epsilon,0])\subset I^{+}([r])$, so that either
$\gamma((-\epsilon,0)) \subset I^{-}([r])$ and hence
$\gamma((-\epsilon,0)) \subset [r]$ a contradiction, or $\epsilon$
could have been chosen so small that $\gamma((-\epsilon,0))\subset
\p I^{-}([r])$ in which case $\gamma((-\epsilon,0))\subset
B_f([r])\backslash B_{p}([r])$ which would contradict the assumption
that $\gamma(0)$ is  the first exit point in the past direction. We
conclude that $\gamma(0)\in B_f([r])\cap B_p([r])$.
\end{proof}

\begin{example}
We construct an example which proves that a generator of  $B_f([r])$
can have starting point belonging to $(B_f([r])\cap
B_p([r]))\backslash \textrm{edge}(B_f([r]))$ and immediately escape $B_f([r])$ if
prolonged in the past direction.

Consider 1+2 Minkowski spacetime of coordinates $(t,x,y)$ and
identify the hyperplanes $t=-2$ and $t=2$, so that the spacetime $N$
between the two slices becomes totally vicious. Next remove from
$t=-1$ an ellipse (including the interior) whose minor axis is $2$
and whose major axis is $4$. Do the same on the slice $t=1$ but let
the new ellipse be rotated of $\pi/2$  radians with respect to the
former. A point belonging to both $B_f([r])\cap B_p([r])$ is
$q=(0,0,0)$; $q$ does not belong to $\textrm{edge}(B_f([r]))$  and
it is easy to check that the two generators starting from $q$ of
$B_f([r])$ escape $B_f([r])$ if prolonged in the past direction.
\end{example}

\subsection{Differentiability of $\dot{[r]}$}

Let us consider the issue of  the differentiability of $\dot{[r]}$.
We regard this set as the union of the two $n$-dimensional topological
submanifolds $B_f([r])\backslash \textrm{edge}(B_f([r]))$ and
$B_p([r])\backslash \textrm{edge}(B_f([r]))$, thus we focus first on
the differentiability of $B_p([r])\backslash
\textrm{edge}(B_p([r]))$.

The differentiability of topological hypersurfaces generated by past inextendible
lightlike geodesics has been studied in
\cite{beem98,chrusciel98,chrusciel98b,minguzzi14d}. This analysis was carried
out having in mind Cauchy horizons but, as it is clarified with
\cite[Theor.\ 2.3]{chrusciel98b}, the results hold in  general. Points at which the generators leave the hypersurface in the
future direction are called future endpoints. The quoted works prove
that at non-future endpoints the hypersurface is $C^1$, at future
endpoints at which ends only one generator the hypersurface is still
$C^1$ and at future endpoints at which ends more than one generator
the hypersurface is non differentiable.  Therefore these results hold  unchanged for
$B_p([r])\backslash \textrm{edge}(B_p([r]))$, and a time dual version holds
 for $B_f([r])\backslash \textrm{edge}(B_f([r]))$. A better way
to apply them is by considering $B_p([r])$ as a subset of $\p
I^{+}([r])$ which is also generated by past inextendible lightlike geodesics. From
that we can infer that $B_p([r])$ is non-differentiable at $p\in
\textrm{edge}(B_p([r]))$ if and only if $p$ admits more that one
generator of $B_p([r])$ ending at it.

Furthermore, Chru\'sciel and Galloway \cite{chrusciel98} have given
an example of Cauchy horizon which is non-differentiable in a dense
set. They first constructed \cite[Theor.\ 1.1]{chrusciel98} a compact
set $C=\mathbb{R}^2\backslash K\subset \mathbb{R}^2$ having a
connected Lipschitz boundary such that on the spacetime $M=(-1,1)
\times \mathbb{R}^2$, endowed with the usual Minkowski metric,
$E^{+}(\{0\}\times C)$ was non-differentiable on a dense set.

We construct an example of spacetime in which $\dot{[r]}$ is
non-differentiable on a dense set as follows. We remove from the
just constructed spacetime the sets $\{0\}\times C$ and
$\{1/2\}\times C$ and we identify the interior of the  upper-side of
$\{0\}\times C$ with the interior of lower-side of $\{1/2\}\times
C$. This operation introduces closed timelike curves and the
boundary of the chronology violating region is a subset of what,
before the removal of the sets, was $E^{+}(\{0\}\times C)\cup
E^{-}(\{1/2\}\times C)$. As such $\dot{[r]}$ is non-differentiable
on a dense set.

We say that  $B_f$ is {\em compactly generated} if there is a compact set $K$ such that its (future inextendible) generators enter  $K$. For the notions of the next theorem not previous introduced we refer the reader to \cite{minguzzi14d}. Observe that for the study of the development of time machines one is interested in the time dual version involving $B_p$.

\begin{theorem} \label{mai}
Assume that the null convergence condition holds. If $B_f$  is compactly
generated and its generators are future complete then  it  is compact,  $C^{3}$, and generated by inextendible lightlike geodesics. Actually smooth if the metric is smooth, and analytic if the metric is analytic. Moreover,  $B_f$ has zero Euler characteristic, it  is generated by future complete lightlike lines and on $B_f$
\[\theta=\sigma^2=Ric(n,n)=0,\qquad b=\overline{\sigma}=\overline{R}=\overline{C}=0.\]
In other words, denoting with $n$ a lightlike tangent field to $B_f$, for every $X\in T B_f$, $\nabla_X n \propto n$ and $R(X, n)n\propto n$, that is, the second fundamental form vanishes on $B_f$ and the null genericity condition is violated everywhere on $B_f$. In 2+1 spacetime dimensions either $B_f$ is a torus or a Klein bottle where the latter case is excluded if the spacetime is time orientable.
\end{theorem}

If it is know that if $B_f=H^{-}(S)$ for some partial Cauchy surface $S$ (e.g.\ see next section) then the condition on the geodesic completeness of $B_f$ can be dropped, for in this case one can use directly \cite[Theor.\ 18]{minguzzi14d}. We stress once again \cite{minguzzi14d} that physically speaking it is incorrect to demand the validity of the null genericity condition on a compact set as done by some authors \cite{tipler77}, thus its violation does not imply that the spacetime is unphysical.

Observe that $B_f$ belongs to the boundary of the chronology violating region, so if it is not compactly generated then the chronology violating region propagates to the boundary of spacetime.
Thus this theorem establishes that either the formation of closed timelike curves happens as in the theorem, with  a compact smooth $B_p$ with all its mentioned nice properties, or such CTC formation either violates energy conditions,  extends to the boundary, or generates (geodesic) singularities.

\begin{proof}
The proof coincides with that of \cite[Theor.\ 18]{minguzzi14d} where this time the completeness of the future hypersurface must be assumed while there it was a consequence of it being a Cauchy horizon. If the spacetime dimension is three $B_f$ is two dimensional and the only compact closed surfaces with Euler characteristic zero are the Klein bottle and the torus. If the spacetime is orientable then $B_f$ is orientable thus the Klein bottle can be excluded.
\end{proof}

In three spacetime dimensions one could obtain other interesting results by applying the Schwartz-{P}oincar\'e-{B}endixson theorem to $B_f$. In fact, observe that $B_f$ is at least $C^3$ thus its tangent vector field is $C^2$ as required by SPB's theorem. We conclude that under the assumption of Theorem \ref{mai} $B_f$ contains a closed causal curve (fountain) which is a minimal invariant set for the future (and past)  lightlike flow on the horizon or the whole torus is itself a minimal invariant set.
We recall that a minimal invariant sets is a closed minimal set which is left invariant by the future flow on $B_f$. The concept makes sense on any imprisoned causal curve. Any minimal invariant set is generated by lightlike lines \cite{kay97,minguzzi07f}.

Any lightlike geodesic is  no more lightlike if we open the light cones. However, in some cases the lightlike geodesic is stable in the sense that it gets simply moved aside, while in other cases it is unstable as it  disappears completely.
The next result does not assume null completeness but, rather, it relates it to the concept of stability.

\begin{theorem}
Suppose that $B_f([r])$ is compactly generated, then any geodesic on $B_f\backslash B_p$ belonging to one of its minimal invariant sets is future complete or it is unstable.
\end{theorem}

A better understanding of this theorem can be obtained from the proof.

\begin{proof}
Suppose that this geodesic $\gamma$ is future incomplete, then there is a small timelike variation towards the future of $\gamma$ which brings this curve to a timelike curve $\eta$ such that $\eta$ accumulates in the future to the same points to which accumulates $\gamma$, and hence accumulates on $\gamma$ itself \cite[Theor.\ 2.1]{minguzzi14} \cite{krasnikov14}. This fact implies that it is possible to construct a closed timelike curve $\sigma$ in $I^{+}(\gamma)\cap U$ where $U$ is any neighborhood of $\gamma$. That is, $\gamma$ is in the past boundary of a chronology violating class $[q]$ and on the future boundary of another chronology violating class $[r]$ (the two classes are different otherwise $\gamma$ would belong to a chronology violating class, just take a timelike curve moving from $[r]$ to $[q]=[r]$ passing through a point of $\gamma$). Thus by opening slightly the cones, $\gamma$ disappears but it cannot be recreated anywhere else since the two distinct classes join in a single class, thus showing that the previous configuration was unstable.
\end{proof}

On might ask whether the violation of chronology near a point of $B_f([r])$ is a local or global phenomenon. The next result shows that if a minimal invariant set generator is incomplete in the past direction, then closed timelike curves can be found in any neighborhood of the generator.

\begin{theorem}
Let $B_f([r])$ be compactly generated and let $\gamma$ be a lightlike geodesic on $B_f\backslash B_p$ belonging to one of its minimal invariant sets. Then either $\gamma$ is past complete or for every neighborhood $U \supset \gamma$ there is a closed timelike curve in $U\cap I^{-}(\gamma)\cap [r]$.
\end{theorem}

The proofs goes similarly to that of the previous theorem but reworked in the past direction.

\subsection{The coincidence with previous definitions of boundary}

The next result shows that, provided the chronal region is globally
hyperbolic, the past Cauchy horizon of a suitable hypersurface is
the future boundary of the chronology violating set. This  result
relates our definition of boundary with the more restrictive one
given in some other papers \cite{thorne93}.

\begin{proposition}
Let $[r]$ be a chronology violating class and assume that the
manifold $N=M\backslash \overline{I^{-}([r])}$ with the induced
metric is globally hyperbolic, then for every Cauchy hypersurface
$S$ of $N$, $S$ is edgeless in $M$ and $H^{-}(S)=\partial
{I^{-}([r])}$. Moreover, if $M=I^{+}([r])$ then
$H^{-}(S)=\dot{[r]}=B_f([r])$.
\end{proposition}

\begin{proof}
Since $S$ is a (acausal) Cauchy hypersurface for $N$, $\textrm{Int}
D(S)=N$, thus $\partial D(S) \subset \dot{N}=\partial {I^{-}([r])}$.
The set $S$ has no edge in $N$, moreover, it has no edge also in $M$.
Indeed, let $q \in \textrm{edge}(S)$, then as $S$ is closed in $N$,
$q \in \partial {I^{-}([r])}$. But $I^{+}(q)$ is an open set that
cannot intersect the past set  $I^{-}([r])$, thus $I^{+}(q) \subset
N$, moreover no inextendible timelike curve starting from $q$ (e.g.\
a geodesic) can intersect $S$ for otherwise $S$ would not be
achronal. But since such curve would be inextendible in $N$ this
would contradict the fact that $S$ is a Cauchy hypersurface. Thus
$\textrm{edge}(S)=\emptyset$.

Note that $\partial D(S)=H^{+}(S)\cup H^{-}(S)$, thus $H^{-}(S)
\subset \p{I^{-}([r])}$. For the converse note that if $p \in
\partial {I^{-}([r])}$, $I^{+}(p)$ is an open set that cannot intersect
$I^{-}([r])$, thus $I^{+}(p) \subset N$. Note that $p \in I^{-}(S)$
for otherwise a future inextendible timelike curve issued from $p$
would not intersect $S$, still when regarded as an inextendible
curve in $N$ this empty intersection would contradict the fact that
$S$ is a Cauchy hypersurface. Since $p \in I^{-}(S)$  the points in
$I^{+}(p) \cap I^{-}(S)$ necessarily belong to $D^{-}(S)$ thus $p
\in \overline{D^{-}(S)}$ and moreover $p$ does not belong to
$\textrm{Int} D^{-}(S)$ because the points in $I^{-}([r])$ clearly
do not belong to $D^{-}(S)$, as the future inextendible timelike
curves issued from there may enter the chronology violating set
$[r]$ and remain there confined. Thus $p \in H^{-}(S)$.

By the previous result if $M=I^{+}([r])$ then $I^{-}([r])=[r]$ and
$B_f([r])=\dot{[r]}$.
\end{proof}

\section{Relationship between compact generation and compact construction} \label{cop}

We have recalled that theorems on the non-existence of time machines are based on the observation that any creation of a region of chronology violation would lead to a Cauchy horizon which is compactly generated, namely, such that the generators followed in the past direction enter and get imprisoned in a compact set $K$. The idea is that the information on the production of closed timelike curves would propagate on spacetime along the generators of the horizon, so followed in the backward direction those generators have to enter the compact space region were the advanced civilization produced the timelike curves in the first place.

This is the argument which is used to justify the assumption of `compact generation of the horizon'  in connection to the creation of time machines.  It must be remarked that the generators being confined to the horizon cannot reach the Cauchy hypersurface, however, they do intersect the world tube of the compact region in which they are past imprisoned. In this sense the term `space region' used in the previous paragraph is appropriate. Nevertheless, Amos Ori in a series of papers \cite{ori05,ori07}  has criticized the previous argument maintaining that the
assumption of local time machine creation would have to be expressed by the
following concept, which he terms compact construction.

\begin{definition}
The Cauchy horizon $H^+(S)$ is {\em compactly constructed} if there is a compact set $ S_0\subset S$ such that $H^+(S_0)\cap H^+(S)$ contains almost closed causal curves.
\end{definition}
Here $S_0$ represents the region were the actions of the advanced civilization leading to the formation of closed timelike curves took place.

Actually, Ori uses ``closed causal curves'' in place of `almost closed causal curves' in the above definition. The difference does not seem to be important: the idea is that almost  closed (and possibly closed) causal curves would signal the development of closed timelike curves just behind the horizon. Ori shows that
a compactly constructed time machine can be initiated with no violation of energy conditions \cite{ori07}.

The relative strength of compact generation and compact construction has remained open so far.
One could suspect `compact construction' to be a weaker property than `compact generation', for the latter with its strength  prevents the formation of time machines. In fact we are able to prove

\begin{theorem}
 Let $S$ be a closed acausal hypersurface without edge (partial Cauchy hypersurface). If $H^+(S)$ is compactly generated then it is compactly constructed.
\end{theorem}

\begin{proof}
Let $K$ be the imprisoning compact, we can assume that $K\subset H^+(S)$, otherwise replace $K$ with $K\cap H^+(S)$. Let $C\subset I^+(S)$  be another compact set, chosen so that $K\subset \textrm{Int}\, C$.
Let us prove  that $S_0:=J^{-}(C\cap \overline{ D^{+}(S)})\cap S$ is compact.  Suppose  not, then there is a sequence of past inextendible casual curves $\gamma_n$ with future endpoint $p_n\in C\cap \overline{ D^{+}(S)}$ which intersect $S$ at $q_n$ with $q_n \to \infty$, meaning by this that the sequence $q_n$ escapes every compact subset of $S$. Following $\gamma_n$ in the future direction let  $r_n \in \p C \cap \overline{ D^{+}(S)} $ be the first point in $C$ and let $\eta_n:=\gamma_n\vert_{q_n \to r_n}$ be the portion of $\gamma_n$ not in $C$ saved for $r_n$. Let $r\in \p C \cap \overline{ D^{+}(S)}$ be an accumulation point of $r_n$. By the limit curve theorem \cite{minguzzi07c} there is a past inextendible causal curve $\eta$ with future endpoint $r$ which does not intersect $S$ (if it were to intersect it at some $y\in S$ then a subsequence $q_{n_s}$ would converge to $y$ which is impossible since every subsequence escapes all compact sets). Being $\eta$ the limit of curves contained in the closed set $M\backslash \textrm{Int}\, C$ it is also contained in this closed set and so does not intersect $K$. Observe that it is a causal curve which cannot enter $D^+(S)$ for otherwise it would be forced to reach $S$, thus it is entirely contained in $H^+(S)$. This fact proves that $r\in H^+(S)$.  Since the horizon is achronal $\eta$ is a lightlike geodesic, that is a generator (lightlike geodesics on the horizon cannot cross for it is easy to see that it would contradict achronality).
This is a contradiction with compact generation since we have shown that $\eta$ does not intersect $K\subset \textrm{Int} C$ where every generator should enter. The contradiction proves that $S_0$ is compact. Let $x\in K$, and consider a sequence $ x_k \to x$, $x_k \in I^{-}(x)$. As a consequence, $x_k \in D^+(S)$. For sufficiently large $k$, $x_k \in C$ which implies that $x_k \in D^+(S_0)$ and consequently, $x\in \overline{D^+(S_0)}$ which implies $x\in H^+(S_0)$. We have shown  that $K\subset H^+(S_0) $ where $K$ contains almost closed causal curves since it contains a minimal invariant set \cite{minguzzi07f}.

%
\end{proof}

\section{The case $I^{+}([r])=M$ and a singularity theorem}

S. Hawking has suggested  that the laws of physics
prevent the formation of closed timelike curves in spacetime
\cite{hawking92} (the chronology protection conjecture). According
to this conjecture the effects preventing the formation of closed
timelike curves could be quantistic in nature, in fact Hawking
claims that the divergence of the stress energy tensor at the
boundary of the chronology violating set would be a feature of this
prevention mechanism.

Despite some work aimed at proving the chronology protection
conjecture its present status remains quite unclear with some papers
supporting it and other papers suggesting its failure
\cite{tipler77,visser96,li96,li98,krasnikov02}. Some people think
that in order to solve the problem of the chronology protection
conjecture a full theory of quantum gravity would be required
\cite{hawking92,gott98}.

A weak form of chronology protection would forbid the formation of
closed timelike curves without denying the possibility that closed
timelike curves could have been present since the very beginning of
the universe. For this reason it is important to study spacetimes
that originate causally from a chronology violating region $[r]$,
namely $I^{+}([r])=M$.
%

\begin{proposition}
There is at most one  chronology violating class $[r]$ with the
property $I^{+}([r])=M$.
\end{proposition}

\begin{proof}
Let $[x]$ be a second chronology violating class such that
$I^{+}([x])=M$ then $x\ll r$ and, since $I^{+}([r])=M$, $r \ll x$
thus $[x]=[r]$.
\end{proof}

\begin{proposition} \label{kjh}
Let $[r]$ be a chronology violating class such that $I^{+}([r])=M$,
then $\dot{[r]}=B_f([r])$, $J^{-}(\overline{[r]}\,)=\overline{[r]}$
and $I^{-}(\overline{[r]}\,)=[r]$.
\end{proposition}

\begin{proof}
Since $I^{+}([r]) \cap \dot{[r]} \subset B_f([r])$ we have
$\dot{[r]} \subset B_f([r])$ and hence the first equality. For the
second equality the inclusion  $\overline{[r]} \subset
J^{-}(\overline{[r]})$ is obvious. For the other direction  assume
by contradiction, $p \in
J^{-}(\overline{[r]})\backslash\overline{[r]}$. Since $p \in
M=I^{+}(r)$ there is a timelike curve joining $r$ to $p$ and a
causal curve joining $p$ to $\overline{[r]}$. By making a small
variation starting near $p$ we get a timelike curve from $r$ to
$\overline{[r]}$, and hence equivalently, from $r$ to $r$ passing
arbitrarily close to $p$, thus $p \in \overline{[r]}$, a
contradiction.

For the last equality it suffices to take the interior of the second
one.
\end{proof}

\begin{proposition} \label{kfh}
Let $[r]$ be a chronology violating class such that $I^{+}([r])=M$.
A past or future inextendible achronal causal curve on $M$ is either
entirely contained in $M\backslash \overline{[r]}$ or in
$\dot{[r]}$.
\end{proposition}

\begin{proof}
Let $\gamma$ be a past inextendible achronal causal curve  which
passes through a point $p\in M\backslash \overline{[r]}$. Let us
follow it to the past of $p$. If it intersects $\dot{[r]}$ at some
point $q$ then it cannot be tangent to a generator $\eta$ of
$B_f([r])$ at $q$, for otherwise it would coincide with that
generator  to the future of $q$ and hence would be entirely
contained in $B_f([r])\subset \overline{[r]}$, a contradiction with
$p\in M\backslash \overline{[r]}$. However, if it makes a corner
with $\eta$ then any point $q' \in \gamma$ to the past of $q$ would
belong to $I^{-}(\overline{[r]})=[r]$, which is impossible since a
lightlike line cannot intersect the chronology violating region.

Let $\gamma$ be a future inextendible achronal causal curve  which
passes through a point $p\in M\backslash \overline{[r]}$. Then it
cannot intersect $\overline{[r]}$ because
$J^{-}(\overline{[r]})=\overline{[r]}$.
\end{proof}

The interesting fact is that  $M\backslash \overline{[r]}$ must
admit a time function, provided null geodesic completeness and other
reasonable physical conditions are satisfied (see Theorem
\ref{vgs}). For more details on these conditions see
\cite{hawking73}. It can be read as a singularity
theorem: under fairly reasonable physical conditions if the
spacetime outside the chronology violating region does not admit a
time function then the spacetime is geodesically singular.

Theorem \ref{vgs} is a non-trivial generalization over the main
theorem contained in \cite{minguzzi07d}. Note that null geodesic
completeness is required only on those geodesics intersecting
$M\backslash\overline{[r]}$. These geodesics cannot be tangent to
some geodesic generating the boundary $\dot{[r]}$, because since
this boundary is generated by future lightlike rays contained in
$\dot{[r]}$ (Prop. \ref{kjh}) the geodesic would have to be
contained in $\overline{[r]}$, a contradiction.

\begin{theorem} \label{vgs}
Let $(M,g)$ be a spacetime which admits  no chronology violating
class but possibly for the one, denoted $[r]$, which generates the
whole universe, i.e.\ $I^{+}([r])=M$. Assume that the spacetime
satisfies the null convergence condition and the null genericity
condition on the lightlike inextendible geodesics which are entirely
contained in $M\backslash\overline{[r]}$, and suppose that these
lightlike geodesics are complete. Then the spacetime $M\backslash
\overline{[r]}$ is  stably causal and hence admits a time function.
\end{theorem}

\begin{proof}
 Consider the spacetime $N=M\backslash
\overline{[r]}$ with the induced metric $g_N$, and denote with
$J^{+}_N$ its causal relation. This spacetime is clearly
chronological and in fact strongly causal. Indeed, if strong
causality would fail at $p \in N$ then there would be sequences
$p_n, q_n \to p$, and causal curves $\sigma_n$ of endpoints $p_n,
q_n$, entirely contained in $N$, but all escaping and reentering
some neighborhood of $p$. By an application of the limit curve
theorem \cite{minguzzi07c,beem96} on the spacetime $M$ there would
be an inextendible continuous causal curve   $\sigma$ passing
through $p$ and contained in $\bar{N}$ to which a reparametrized
subsequence $\sigma_n$ converges uniformly on compact subsets
($\sigma$ can possibly be closed). The curve $\sigma$  must be
achronal otherwise one would easily construct a closed timelike
curve intersecting $N$ (a piece of this curve would be a segment of
some $\sigma_n$ thus intersecting $N$). Thus $\sigma$ is a lightlike
line and hence, by Lemma \ref{kfh}, it  is entirely contained in
$N$. By assumption $\sigma$ is complete thus by null genericity and
null convergence  it has conjugate points, which is in contradiction
with it being achronal. The contradiction proves that $(N,g_N)$ is
strongly causal.



 The next step is to prove  that $\overline{J^{+}_N}$
is transitive. In this case $N$ would be causally easy
\cite{minguzzi08b} and hence stably causal (thus admitting time
functions). Suppose $(x,y)\in \overline{J^{+}_N}$ and $(y,z)\in
\overline{J^{+}_N}$. The transitivity of $\overline{J^{+}_N}$ is
proved as done in \cite[Theorem 5]{minguzzi07d}, observing  that the
limit curve passing through  $y$ constructed in that proof,
necessarily contained in $\bar{N}$, is either achronal and hence, by
Lemma \ref{kfh}, entirely contained in $N$, which allows to apply
that original argument, or non-achronal. In the latter case that
argument of proof shows that $(x,z)\in \overline{J^+}$. Let us
recall that $\overline{J^+}=\overline{I^+}$, thus there are
neighborhoods $U$ and $V$ such that any timelike curve connecting
$U\ni x$, $U \subset N$ to $V \ni z$, $V \subset N$ must stay in
$N$, because otherwise there would be some $w \in \overline{[r]}$
such that $x' \le w$, with $x' \in U$. This is impossible because by
Prop.\ \ref{kjh}, $J^{-}(\overline{[r]}\,) \subset
\overline{[r]}$. Thus $(x,z)\in \overline{I^+_N}=\overline{J^+_N}$.
%
%
%
%
\end{proof}

In a different work \cite{minguzzi09d} I have argued, using entropic and homogeneity
arguments, that our spacetime could indeed have been causally
preceded by a region of chronology violation. In this picture the
{\em null} hypersurface $\dot{[r]}$ would be generated by achronal
inextendible lightlike geodesics, and would replace the usual
 Big Bang (which is usually taken as a {\em spacelike} hypersurface in the spacetime completion). Since $\dot{[r]}$ would be
generated by lightlike lines a rigidity mechanism would take place
and several components of the Weyl tensor would  vanish at the
boundary (because the Weyl tensor causes focusing \cite{hawking73}).
This fact is in accordance with Penrose's expectations on the
beginning of the universe \cite{penrose79} (the Weyl tensor
hypothesis) according to which, in order to solve the entropic
problem of cosmology, the Weyl tensor must be small at the beginning
of the Universe.

\section{Conclusions}

We have studied the boundary of the chronology violating set, defining  its future and past parts and proving the reasonability of the definition. For instance, we have shown that the edges of these parts coincide and that the full boundary is obtained by gluing the future and past parts along their edges. We have shown that our definitions are compatible with a previous  definition in the  domain of applicability of the latter.
 We have studied other properties of these boundaries, including causal convexity, differentiability and smoothness under energy conditions. Theorem \ref{mai} clarified the connection with singularities.  We have also proved that compactly generated horizons are compactly constructed. This results did not use the definition of chronological boundary but it is relevant in order to clarify no-go theorems on the  creation of time machines.
 Finally, we have considered the circumstance in which there is just one chronology violating region at the beginning of the Universe, proving that under reasonable energy and genericity conditions either there is a time function outside it or the spacetime is singular.

\section*{Acknowledgments}
I wish to thank Amos Ori for the useful comments and the stimulating discussions. This work has been partially supported by GNFM of INDAM.\\



\end{document}